\documentclass[reqno]{amsart}
\usepackage{graphicx}

\newtheorem{lemma}{Lemma}
\newtheorem{theorem}{Theorem}
\newtheorem{proposition}{Proposition}
\newtheorem{corollary}{Corollary}

\newtheorem{remark}{Remark}

\newcommand{\bee}{\begin{eqnarray}}
\newcommand{\eee}{\end{eqnarray}}
\newcommand{\be}{\begin{eqnarray*}}
\newcommand{\ee}{\end{eqnarray*}}
\newcommand{\R}{{\mathbb R}}

\newcommand{\En}{{\mathcal E}}

\newcommand{\I}{{\mathcal I}}
\newcommand{\V}{{\mathcal V}}
\newcommand{\Z}{{\mathcal Z}}
\newcommand{\T}{{\mathcal T}}
\newcommand{\No}{{\mathcal N}}
\newcommand{\z}{{\zeta}}

\begin{document}
 
 \title [Zakharov-Glassey enhancement]{Enhancement of the Zakharov-Glassey's method for Blow-Up in nonlinear Schr\"odinger equations}
 
 \author {
 Andrea Sacchetti
 }

\address {
Department of Physics, Informatics and Mathematics, University of Modena and Reggio Emilia, Modena, Italy.
}

\email {andrea.sacchetti@unimore.it}

\date {\today}

\thanks {The author thanks R.Carles and C.Sparber for helpful comments. \ This work is partially supported by the GNFM-INdAM and the UniMoRe-FIM 
project  ``Modelli e metodi della Fisica Matematica''.}

\begin {abstract} In this paper we give a sharper condition for blow-up of the solution to a nonlinear Schr\"odinger equation with free/Stark/quadratic 
potential by improving the well known Zakharov-Glassey's method. 

PACS number(s): {05.45.-a, 03.65.-w, 03.65.Db, 03.75.Lm}

MSC 2020 number(s): {35Qxx, 81Qxx}
\end{abstract}


\maketitle

\section {Introduction}

We consider, in dimension one, the nonlinear Schr\"odinger equation
\bee
\left \{
\begin {array}{l}
i \hbar \frac {\partial \psi_t}{\partial t} = H\psi_t + \nu |\psi_t |^{2\mu } \psi_t   \\
\left. \psi_t (x)\right |_{t=t_0} =\psi_0 (x)\, , \ \| \psi_0 \|_{L^2} = 1 \, , 
\end {array}
\right.
\, , \ \psi_t \in L^2 (\R , dx)\, , \label {Eq1}
\eee
where $H = -  \frac {\hbar^2}{2m} \frac {\partial^2 }{\partial x^2} + V (x)$ is the linear Schr\"odigner operator with potential $V(x)$; $\nu \in \R$ 
represents the strength of the nonlinear perturbation and $\mu >0$ is the nonlinearity power. \ Hereafter, for sake of simplicity, we fix the units such 
that $\hbar=1$ and $m=1$, we further assume that $t_0=0$. \ The restriction to dimension one is just to simplify the discussion, but extension to higher 
dimensions of the ideas in this paper could be possible; however, we do not dwell on this problem here. 

The first fundamental question that arises when dealing with a nonlinear Schr\"o\-din\-ger equation (\ref {Eq1}) is the existence of a solution locally in time 
in some functional space. \ Thus, for $\psi_0$ in such a space and under some assumptions on the potential $V(x)$, there exists $0 < t_+^\star \le +\infty$ such 
that $\psi_t \in C([0, t_+^\star ) )$; furthermore, conservation of the norm
\bee
\No (\psi_t ) =\No (\psi_0) \, \ \mbox { where } \ \No (\psi ):= \| \psi \|_{L^2} \, , \label {Eq2}
\eee
and of the energy
\bee
\En (\psi_t ) =\En (\psi_0) \, \ \mbox { where } \ \En (\psi ):= \left \langle \psi , H \psi \right \rangle_{L^2} + 
\frac {\nu}{\mu +1} \| \psi \|_{L^{2\mu +2}}^{2\mu +2}  \, , \label {Eq3}
\eee
are satisfied. \ Concerning global existence in the future three possibilities may occur:

\begin {itemize}

\item [-] $t_+^\star = +\infty$ and $\limsup_{t\to + \infty} \| \psi_t \|_{H^1} < +\infty$, that is the solution is global and bounded;

\item [-] $ t_+^\star =+\infty$ and $\limsup_{t\to + \infty} \| \psi_t \|_{H^1} = +\infty$, that is the solution blows up in infinite
time;

\item [-] $0 < t_+^\star < +\infty$ and $\| \psi_t\|_{H^1}\to + \infty$ as $t \to t_+^\star-0$, that is the solution blows up in finite time.

\end {itemize}

A similar analysis can be considered in the past for $t\le 0$.

Our pourpose in this paper is to give a blow-up criterion by improving the Zakharov(-Shabat)-Glassey's method. \ The method introduced by Zakarhov and 
Shabat \cite {ZS} and by Glassey \cite {G} (see also the papers by \cite {Ka,M,R}) is quite simple in the case where the virial identity takes a simple form. 

Let
\be
\I (t) = \langle \psi_t , x^2 \psi_t \rangle_{L^2}
\ee
be the \emph {moment of inertia}. \ It should be noticed that some textbooks denote (improperly) $\I$ by the name of \emph {variance}; in fact, we will introduce 
the \emph {variance} later. \ Eventually, in analogy with the usual definition in Classical Mechanics the term $\I$ should be properly called \emph {(polar) 
moment of inertia}.

If it can be shown that $\I (T_+^\I)=0$ (resp. $\I (T_-^\I)=0$) for some $\pm T_{\pm}^\I >0$ then blow-up occurs in the future at some $t^\star_+ \in (0,T_+^\I]$ 
(resp. in the past at some $t^\star_- \in [T_-^\I,0)$). \ This fact is a consequence of the functional inequality (\ref {Eq28}) and of the conservation of the 
norm (\ref {Eq2}). \ In order to prove that $\I (t)$ can take zero value at some instant $t$ one usually makes use of the virial identity, which in the one-dimensional 
free model where $V\equiv 0$ takes the form 
\bee
\frac {d^2 \I}{d t^2} = C_{\I} +  {2\nu} \frac {\mu -2}{\mu +1} \| \psi_t \|_{L^{2\mu +2}}^{2\mu +2} \, , \ C_{\I} =  {4 {\En}(\psi_0)} \, . \label {Eq4}
\eee
If, for example, $\mu=2$ and $\psi_0 $ is such that $\En (\psi_0 ) < 0$, then by the virial identity (\ref {Eq4}) and by the conservation of the energy, the 
positive quantity $\I (t) $ is an inverted parabola that must then become negative in finite times $T_\pm^\I$, $-\infty < T^I_- < 0 < T_+^\I < +\infty$, and thus 
the solution cannot exist for all time and blows up at finite time in the future as well as in the past \cite {OT}. 

This argument is very powerful because of its simplicity, in fact it is based on a pure Hamiltonian information $\En (\psi_0) < 0$, and it also applies to the super-critical case $\mu > 2$. \ On the other hand, it strongly depends on the virial identity (\ref {Eq4}) and thus it cannot simply be applied when an external potential $V(x)$ is present. \ However, in a sequence of seminal papers by Carles \cite {Car1,Car2,Car4,Car3} this method has been applied to the case where $V(x)$ is a quadratic or Stark potential in any dimension.

Our proposal of enhancement of the Zakarov-Glassey's method is based on a quite simple idea. \ Let
\bee 
\langle \hat x \rangle^t := \left \langle \psi_t , x \psi_t \right \rangle_{L^2} \label {Eq5}
\eee
be the expectation value of the position observable $x$, where $\hat x$ is the  associated operator. \ Let 
\be
\V (t) = \left \langle \psi_t , \left ( \hat x - \langle \hat x \rangle^t \right )^2 \psi_t \right \rangle = I(t) - \left [ \langle \hat x \rangle^t \right ]^2 
\ee
be the \emph {variance}. \ If it can be shown that $\V (T_+^\V)=0$ (resp. $\V (T_-^\V)=0$) at some $\pm T_{\pm}^\V >0$ then blow-up occurs in the future for some 
$t^\star_+ \in (0,T_+^\V]$ (resp. in the past for some $t^\star_- \in [T_-^\V,0)$), by (\ref {Eq29}. \ Since $\V (t) \le \I(t)$ then we expect to give a sharper 
condition for the occurrence of the blow-up; the price to pay is to give an expression of the expectation value $\langle x \rangle^t$, but this problem can 
be easily overcome using the (generalized) Ehrenfest's Theorem where $\langle x \rangle^t$ is nothing but the solution of the ``classical mechanics 
equation''. \ Finally, we must also emphasize the fact that the enhanced Zakharov-Glassey's method not only gives sharper conditions for the occurrence of 
blow-up but also allows us to give a better estimate of the instants $t_\pm^\star$ at which the solution becomes singular because $|T_\pm^\V |\le |T_\pm^\I |$.

The paper is organized as follows. \ In Section \ref {Sez2} we recall the Ehrenfest's generalized Theorem; in Section \ref {Sez3} we review the standard blow-up 
conditions in the free model where $V(x) \equiv 0$ and we show that these conditions can be easily improved by applying the virial equation for the variance $\V (t)$; in 
Section \ref {Sez4} we consider the case where $V(x) = \alpha x$, $\alpha \in \R$, is a Stark potential; in Section \ref {Sez5} we review the blow-up 
conditions in the case where $V(x) = \alpha x^2$, $\alpha \in \R$, is a quadratic potential and we show that again these conditions can be easily improved by applying 
the virial equation for the variance $\V (t)$. \ Finally, Appendix \ref {AppA} is about some functional inequalities, Appendix \ref {AppB} is about a comparison result 
for ordinary differential equations and Appendix \ref {AppC} is about the formal derivation of the virial identity; some results in Appendices \ref {AppB} 
and \ref {AppC} are due to the papers \cite {Car1,Car2,Car4,Car3}, I collect these results in two short Appendices for reader's benefit.

Hereafter, for the sake of simplicity, we omit the dependence on the variable $t$ when this fact does not cause misunderstandings, e.g. $\psi$ instead of 
$\psi_t$, $\langle \hat x\rangle$ instead of $\langle \hat x \rangle^t$, $\langle \hat p\rangle$ instead of $\langle \hat p \rangle^t$, $\I$ instead of 
$\I (t)$, $\V$ instead of $\V (t)$, and so on. 

By $f' = \frac {df}{dx}$ we denote the derivative with respect to $x$, by $\langle f,g \rangle_{L^2 }$ we denote the scalar product $\int_{\R} \bar f(x) g(x) dx$, 
and it is sometimes denoted simply by $\langle f,g \rangle$; also $\| f \|$ sometimes simply denotes $\| f \|_{L^2 }$.

\section {Ehrenfest's generalized Theorem for NLS} \label {Sez2}

The extension of the Ehrenfest's Theorem to the nonlinear Schr\"odinger equation (\ref {Eq1}) has already been considered by \cite {B,K}. \ In fact, by means of a 
straightforward calculation it follows that

\begin {proposition} \label {Prop1}
Let $a =a (x,p)$, $x,p \in \R$, be a classical observable function with associated operator $A$, let 
\bee
\langle A \rangle = \langle \psi_t, 
A \psi_t \rangle_{L^2 } \label {Eq6}
\eee
be its expectation value. \ Then
\bee
\frac {d\langle  A \rangle}{dt} = i  
\left \langle \psi_t , \left [ H , A \right ] \psi_t \right \rangle_{L^2 } 
+ i \nu 
\left \langle \psi_t , 
\left [ |\psi_t |^{2\mu } , A \right ] \psi_t 
\right \rangle_{L^2 } 
 \, , \label {Eq7}
\eee
 where $[H,A]=HA-AH$ is the commutator operator between the operators $H$ and $A$, and where $\left [ |\psi |^{2\mu } , A \right ] \psi  = 
 |\psi |^{2\mu } A (\psi ) - A( |\psi |^{2\mu } \psi )$. \ Equation (\ref {Eq7}) is usually called ``Ehrenfest's generalized Theorem''.
\end {proposition}

As a consequence it follows that

\begin {corollary} \label {Coro1}
Let $x$ be the position observable and let $\hat x = x$ be the associated multiplication operator, then
\bee
\frac {d\langle  \hat x \rangle^t}{dt} =   \langle \hat p \rangle^t \label {Eq8}
\eee
where $\hat p = -{i} \frac {\partial }{\partial x}$ is the associated operator to the momentum observable $p$. 
\end {corollary}

\begin {proof}
Corollary \ref {Coro1} immediately follows from (\ref {Eq7}) since $ [ |\psi |^{2\mu } , \hat x  ] =0$; hence
\be
\frac {d\langle  \hat x \rangle}{dt} =  {i} 
\left \langle \psi , \left [ H , \hat x \right ] \psi \right \rangle = {i} 
\left \langle \psi , \left [ \frac {\hat p^2}{2} , \hat x \right ] \psi \right \rangle =  \langle \hat p \rangle \, .
\ee
\end {proof}

Similarly

\begin {corollary} \label {Coro2}
Let $p$ be the momentum observable with associated operator  $\hat p = -{i} \frac {\partial }{\partial x}$, then
\bee
\frac {d\langle  \hat p \rangle^t}{dt} = -\left \langle \frac {dV}{dx} \right \rangle^t \, , \ \mbox { where } 
\left \langle \frac {dV}{dx} \right \rangle^t = \left \langle \psi_t ,\frac {dV}{dx}  \psi_t \right \rangle_{L^2 } . \label {Eq9}
\eee
\end {corollary}

\begin {proof} Corollary \ref {Coro2} follows from (\ref {Eq7}) if we prove that $ \left \langle \psi , \left [ |\psi |^{2\mu } , \hat p \right ] \psi 
\right \rangle  =0$; indeed
\be
&& \left \langle \psi , \left [ |\psi |^{2\mu } , \hat p \right ] \psi 
\right \rangle  = -{i} \int_{\R} \bar \psi \left [ |\psi |^{2\mu} \frac  {\partial \psi}{\partial x} - \frac  {\partial \left ( |\psi |^{2\mu} \psi \right )}{\partial x} 
\right ] dx =  \\ 
&& \ \ =  -{i} \int_{\R} |\psi |^{2\mu}  \left [ \bar \psi \frac  {\partial \psi}{\partial x} + \psi \frac  {\partial \bar \psi}{\partial x} \right ] dx = -{i} \int_{\R} 
\rho^{\mu}   \frac  {\partial \rho}{\partial x}  dx = 0  
\ee
where $\rho = |\psi |^2$. \ Hence
\be
\frac {d\langle  \hat p \rangle}{dt} = {i} 
\left \langle \psi , \left [ H , \hat p \right ] \psi \right \rangle =  {i} 
\left \langle \psi , \left [ V , \hat p \right ] \psi \right \rangle = -\left \langle \frac {dV}{dx}  \right \rangle \, . 
\ee
\end {proof}

\begin {remark} \label {Nota1}
Let $\psi_t$ be the solution to the NLS (\ref {Eq1}); then the expectation values $\langle \hat x \rangle$ of the position observable and $\langle \hat p \rangle$ of the 
momentum observable  satisfy to the ``classical canonical equation of motion'' (\ref {Eq8}-\ref {Eq9}). \ In the case where $V(x)$ is a free, Stark or quadratic 
potential then the system (\ref {Eq8}-\ref {Eq9}) has an explicit solution that does not depend on the nonlinearity parameter $\nu$.
\end {remark}

\begin {remark} \label {Nota1Bis} 
Ehrenfest's generalized Theorem was also proved by \cite {ANS} for nonlinear Schr\"odinger equations with a 2 or 3-dimensional confining  harmonic potential and under the 
effect of a rotating force. \ In such a framework it has also been proved that, under some circumstances (see Proposition 4.3 by \cite {ANS}), the solution is such 
that $\langle \hat x \rangle^t$ and $\langle \hat p \rangle^t$ go to $+\infty$ when $t$ goes to $\pm \infty$.
\end {remark}

\begin {remark} \label {Nota2}
However, we should point out that the Ehrenfest's generalized Theorem (\ref {Eq7}) for nonlinear Schr\"odinger does not give the same result of the usual one for 
linear Schr\"odinger equations
\bee
\frac {d \langle A \rangle }{dt} = {i} \left \langle \left [ H,A \right ] \right \rangle \label {Eq10}
\eee
if the classical observable is the Hamiltonian function $h (x,p) = \frac {1}{2}p^2 + V (x)$ with associated operator $H$; indeed, in such a case 
\be
\frac {d\langle  H \rangle}{dt} =  {i} \nu \left \langle \psi , \left [ |\psi |^{2\mu} , H \right ] \psi \right \rangle =  -  {\nu} \Im \left \langle \psi , 
|\psi |^{2\mu } \hat p^2 \psi 
\right \rangle 
\ee
is not generically zero. \ In fact, $\langle H \rangle$ is an integral of motion only when $\nu =0$; otherwise the integral of motion is the energy ${\En} (\psi)$ 
defined by (\ref {Eq3}).
\end {remark}

\section {Blow-up for the free NLS} \label {Sez3} We consider now the case where the external potential is zero: $V(x) \equiv 0$. We assume that 
\bee
\psi_0 \in \Sigma := H^1 (\R ) \cap {\mathcal D} (\hat x)\, , \label {Eq11}
\eee
where ${\mathcal D} (\hat x)$ is the domain of the operator $\hat x$. \ Then the solution $\psi (x,t)$ to (\ref {Eq1}) locally exists and it belongs to 
$C((t^\star_- ,t^\star_+ ), \Sigma)$ and the conservation of the norm $\| \psi \|_{L^2}$ and of the energy ${\mathcal E}$ hold true (see, e.g., Theorem 3.10 
by \cite {SS}). \ If $t^\star_\pm =\pm \infty$ then the solution globally exists; if not, i.e. $t^\star_+ <+\infty $ (resp. $t^\star_- >-\infty$) then 
\be
\lim_{t\to t^\star_\pm \mp 0} \|  \psi \|_{H^1} = \infty 
\ee
and thus blow-up occurs in the future (resp. in the the past). \ We observe that blow-up cannot occur when $\nu \ge 0$ because of the conservation of the 
energy (\ref {Eq3}). \ Furthermore, we can also point out that when blow-up occurs for $\nu <0$ then we also have that 
\be
\lim_{t\to t^\star_\pm \mp 0} \| \psi \|^{2\mu +2}_{L^{2\mu +2}} = \infty 
\ee
because conservation of the energy.

\subsection {Criterion for blow-up by means of the Zakharov-Glassey method} \label {Sez3_1}
Estimates of the momentum of inertia can be obtained by means of the one-dimensional virial identity (\ref {Eq4}) the with initial conditions
\bee
\I_0 := \I (0) = \| x \psi_0 \|_{L^2}^2 \label {Eq12}
\eee
and 
\bee 
\dot \I_0 := \frac { d\I (0)}{dt} = 2 \Im \left [ \int_{\R} x \bar \psi_0 (x) \frac {\partial \psi_0 (x)}{\partial x} dx \right ] = 2 \Re \left \langle \hat x \psi_0 , 
\hat p \psi_0 \right \rangle \, . \label {Eq13}
\eee

Theorem 5.1 by \cite {SS} gives a condition for blow-up in the future (and similarly in the past). \ Specifically, when $\nu <0$ and $\mu \ge 2$ then there 
exists a $t_+^\star \in (0,+\infty)$ such that 
\be
\lim_{t\to t^\star_+ -0} \|  \psi \|_{H^1} = \infty  
\ee
if any of the following conditions is satisfied:
\begin {itemize}

\item [i.] $C_{\I} <0$;

\item [ii.] $C_{\I}=0$ and $\dot \I_0  <0$;

\item [iii.] $C_{\I}>0$ and $\dot \I_0 \le - \sqrt {2C_{\I} {\I_0 } }$;

\end {itemize}

where $C(\I )= 4 \En (\psi_0 )$.

The proof of Theorem 5.1 by \cite {SS} is quite simple: if $\mu \ge 2$ and $\nu \le 0$ then (\ref {Eq4}) implies that 
\be
\frac {d^2 \I}{dt^2} \le C_{\I}
\ee
and thus
\bee
\I (t) \le M(t):=\frac 12 C_{\I} t^2 + \dot \I_0 t + \I_0 \label {Eq14}
\eee
If any of the three conditions i.-iii. are satisfied then there exists $\tilde T^{\I}_+>0$ such that $M(\tilde T^{\I}_+)=0$ and thus there exists a 
$0< T_+^{\I} <\tilde T^{\I}_+$ such that $\I (T_+^{\I} )=0$. \ From this fact and from (\ref {Eq28}) the occurrence of blow-up in the future follows at some 
$t^\star_+ <T_+^{\I}$.

\subsection {Criterion for blow-up by means of the enhanced Zakharov-Glassey method}\label {Sez3_2}
We improve now the previous criterion by  applying the same argument to the analysis of the variance and making use of the Ehrenfest's generalized Theorem. \ Indeed, 
if the potential $V(x)$ is exactly zero then  (\ref {Eq8}-\ref {Eq9}) imply that 
\bee
\langle \hat p \rangle \equiv \hat p_0    \ \mbox { and } \ \langle \hat x \rangle =  { \hat p_0  }t + \hat x_0 \, , \ \mbox { where } \  
\hat x_0 := \left. \langle \hat x \rangle^t \right |_{t=0} \mbox { and } \hat p_0  := \left. \langle \hat p \rangle^t \right |_{t=0} . \label {ini}
\eee

\begin {remark}
We point out that in the free NLS problem the conservation of the momentum $\langle \hat p \rangle$ and the fact that the \emph {center of mass} of the 
wavepacket $\langle \hat x \rangle$ moves at constant speed can be derived by making use of arguments of invariance of space translation (see, e.g. \S 2.3 by \cite {SS}).
\end {remark}

Since (\ref {ini}) we have that
\be
\V (t) &=& \I (t) - \langle \hat x \rangle^2 = \I (t) - \left [  { \hat p_0  }t +  \hat x_0 \right ]^2  \le N(t)
\ee
where
\bee
N(t) &:= & M(t) - \left [ { \hat p_0  }t + \hat x_0 \right ]^2 \nonumber \\
& = &  \left [ \frac 12 C_{\I} - \hat p_0 ^2 \right ] t^2 + \left [ \dot \I_0  - 2 \hat p_0 \hat x_0 \right ] t + \left [ \I_0 - \hat x_0^2 \right ]  \label {Eq15}
\eee
Thus we have the following improvement of Theorem 5.1 by \cite {SS}.

\begin {theorem} \label {Teo1} Let $\nu <0$ and $\mu \ge 2$, let $\psi_0 \in \Sigma$; then we have blow-up in the future if any of the following conditions is satisfied:

\begin {itemize}
 
\item [i'.] $ C_{\I} < 2 \hat p_0^2$;

\item [ii'.] $ C_{\I} = 2 \hat p_0^2$ and $ \dot \I_0 < 2 \hat p_0 \hat x_0 $;

\item [iii'.] $ C_{\I} > 2 \hat p_0^2$ and 
\be
\left [ \dot \I_0 - 2 
\hat p_0 \hat x_0 \right ] \le - 
2\sqrt { \left [ \frac 12 C_{\I} - 
\hat p_0^2 \right ]\left [ \I_0 - 
\hat x_0^2 \right ] }
\ee
 
\end {itemize}

\end {theorem}

\begin {proof} Indeed, if any of the three conditions i.-iii. are satisfied then there exists $\tilde T^{\V}_+>0$ such that $N(\tilde T^{\V}_+)=0$ and thus there 
exists $0<T_+^{\V}<\tilde T^{\V}_+$ such that $\V (T_+^{\V})=0$. \ From this fact and from (\ref {Eq29}) the occurrence of blow-up follows for some $t^\star_+ \le T_+^{\V}$.
\end {proof}

\begin {remark} \label {Nota3}
In fact, under condition $i'.$ we have blow-up in the future and in the past, too; under conditions $ii'.$ and $iii'.$ we have blow-up in the future only.
\end {remark}

\begin {remark} \label {Nota4}
We remark that condition $i'.$ for blow-up is not new and it has been already proved under some circumstances, see e.g. Corollary 1.2 by \cite {DWZ} and Theorem 7 by \cite {R}.
\end {remark}


\section {Blow-up for the NLS with Stark potential} \label {Sez4}
Let the potential $V(x)=\alpha x$ be a Stark potential, where $\alpha \in \R \setminus \{ 0\}$, the occurrence of blow-up in such a case has been 
considered by \cite {Car4,Li,SZ} . \ Again we assume (\ref {Eq11}). 

\subsection {Criterion for blow-up by means of the Zakharov-Glassey method} \label {Sez4_1}
In the case of Stark potentials it has been proved that the solutions to the NLS (\ref {Eq1}) with a Stark potential can be derived from the ones of the 
free NLS, see Theorem 2.1 by \cite {Car4}. \ Then one can make use of the results obtained in Section \ref {Sez3_1}; in particular, Corollary 3.3 by \cite {Car4} 
states that blow-up occurs in the past and in future when 
\bee
\frac 12 \| \psi_0' \|_{L^2} + \frac {\nu}{\mu +1} \| \psi_0 \|^{2\mu +2}_{L^{2\mu +2}} < 0\, . \label {Eq16Bis}
\eee

\subsection {Criterion for blow-up by means of the enhanced Zakharov-Glassey method} \label {Sez4_2}
If the potential $V(x)=\alpha x$ is a Stark potential, where $\alpha \in \R \setminus \{ 0\}$ then  (\ref {Eq8}-\ref {Eq9}) imply that 
\bee
\langle \hat p \rangle^t = - \alpha t +  \hat p_0  \ \mbox { and } \ \langle \hat x \rangle^t =  - \frac 12 \alpha t^2 +  \hat p_0  t  +  \hat x_0  \label {Eq16}
\eee
where 
\be 
\hat x_0  = \left. \langle \hat x \rangle^t \right |_{t=0} \ \mbox { and } \ \hat p_0 =  \left. \langle \hat p \rangle^t \right |_{t=0}\, . 
\ee
Estimates of the momentum of inertia can be obtained by means of the one-dimensional virial identity (\ref {Eq35}) with initial conditions (\ref {Eq12}-\ref {Eq13}). 

If $\nu (\mu -2) \le 0$ then (\ref {Eq16}) and (\ref {Eq35}) imply that
\be
\frac {d^2 \I}{dt^2} \le 4 \En - 6 \alpha \left (- \frac 12 \alpha t^2 + 
\hat p_0  t + \hat x_0 \right )
\ee
where
\be
\En = \frac 12 \| \psi ' \|_{L^2}^2 + \alpha \langle \hat x \rangle + \frac {\nu}{\mu+1} \| \psi \|_{L^{2\mu+2}}^{2\mu +2} \, , 
\ee
and thus
\be
\I (t) \le \frac 14 \alpha^2 t^4 - \alpha \hat p_0 t^3 + \left [ 2 \En - 3 \alpha  \hat x_0  \right ] t^2 + \dot \I_0 t + \I_0 \, ,
\ee
$\I_0$ and $\dot \I_0$ are given by (\ref {Eq12}) and (\ref {Eq13}). \ Therefore,
\be
\V (t) &=& \I (t) - \left [ \langle \hat x \rangle^t \right ]^2 \\ 
&\le& \left [ \| \psi_0' \|^2 + \frac {2\nu}{\mu+1} \| \psi \|_{L^{2\mu+2}}^{2\mu +2} - \hat p_0^2  \right ] t^2 + 
2\left [ \Re \langle \hat x \psi_0 , \hat p \psi_0 \rangle  - \hat p_0 \hat x_0  \right ] t + \V (0)
\ee

Thus, we can conclude that

\begin {theorem} \label {Teo2}
If
\begin {itemize}

\item [i.] $ \| \psi_0' \|^2 + \frac {2\nu}{\mu+1} \| \psi_0 \|_{L^{2\mu+2}}^{2\mu +2} <   \hat p_0^2 $ then we have blow-up in the past and in the future;

\item [ii.]  $ \| \psi_0' \|^2 + \frac {2\nu}{\mu+1} \| \psi_0 \|_{L^{2\mu+2}}^{2\mu +2} =   \hat p_0^2 $ and $\Re \langle \hat x \psi_0 , \hat p \psi_0 \rangle  - 
\hat p_0 \hat x_0  \not= 0$ we have blow-up in the past or in the future;

\item [iii.]  $ \| \psi_0' \|^2 + \frac {2\nu}{\mu+1} \| \psi_0 \|_{L^{2\mu+2}}^{2\mu +2} >   \hat p_0^2 $ and 
\be
\left [ \Re \langle \hat x \psi_0 , \hat p \psi_0 \rangle  - \hat p_0 
\hat x_0 \right ]^2 > \left [  \| \psi_0' \|^2 + 
\frac {2\nu}{\mu+1} \| \psi_0 \|_{L^{2\mu+2}}^{2\mu +2} - \hat p_0^2  \right ] \V (0)
\ee
we have blow-up in the past or in the future.

\end {itemize}

\end {theorem}

\begin {remark} \label {Nota6}
Since $\| \psi_0' \|^2 = \| \hat p \psi_0 \|^2$ and $|\hat p_0 | = |\langle \psi_0 ,\hat p \psi_0 \rangle | \le \| \hat p \psi_0  \|$ then conditions $i.$ and 
$ii.$ holds true only when $\nu <0$.
\end {remark}

\begin {remark} \label {Nota7}
We remark that the blow-up condition (\ref {Eq16Bis}) given by \cite {Car4} agrees with Theorem \ref {Teo2}, indeed (\ref {Eq16Bis}) implies i..
\end {remark}

\section {Blow-up for NLS with harmonic/inverted oscillator potential} \label {Sez5}

In this section we consider the cases of \emph {harmonic oscillator} potential $V(x) = \alpha x^2$, where $\alpha > 0$, and \emph {inverted oscillator} potential, 
where $\alpha <0$. \ The occurrence of blow-up in these cases has been considered by several authors under different assumptions 
\cite {Car1,Car2,Car3,CG,Jao,SZ2,XL,YLZ,ZA}. 

In this Section we consider the blow-up conditions obtained by means of the enhanced Zakharov-Glassey's method and then we compare these results with the previous 
ones obtained by Carles \cite {Car1,Car2,Car3}.

We require now some preliminary results.

Also in this case assume (\ref {Eq11}), then local in time existence of the solution to (\ref {Eq1}) in $\Sigma$ and conservation of the norm and of the 
energy ${\mathcal E}$ follows (see, e.g., \cite {Car3}).

Corollary \ref {Coro2} implies that $\frac {d{\langle {\hat {p}} \rangle^t}}{dt} =  -2\alpha \langle \hat x \rangle^t$; hence the expectation value of the 
position observable coincides with the classical solution. \ More precisley, let 
\be
\lambda^2 = {2 |\alpha |} \, , \   
\hat x_0 = \left. \langle \hat x \rangle^t \right |_{t=0}  \ \mbox { and } \ 
\hat p_0 = \left. \langle \hat p \rangle^t \right |_{t=0} ;
\ee
then the Ehrenfest's generalized Theorem implies that:

\begin {itemize}

\item [-] In the case of the \emph {harmonic oscillator} where $\alpha  >0$, then
\bee
\left \{
\begin {array}{lcl}
\langle \hat x \rangle^t &=& \hat x_0 \cos (\lambda t) + \frac {\hat p_0}{ \lambda}  \sin (\lambda t) \\ 
\langle \hat p \rangle^t &=& - \lambda \hat x_0 \sin (\lambda t) + \hat p_0 \cos (\lambda t) 
\end {array}
\right. \, . \label {Eq17}
\eee

\item [-] In the case of the \emph {inverted oscillator} where $\alpha <0$, then
\bee
\left \{
\begin {array}{lcl}
\langle \hat x \rangle^t &=& \hat x_0 \cosh (\lambda t) + \frac {\hat p_0}{ \lambda}  \sinh (\lambda t) \\ 
\langle \hat p \rangle^t &=& \lambda \hat x_0 \sinh (\lambda t) + \hat p_0 \cosh (\lambda t) 
\end {array}
\right. \, . \label {Eq18}
\eee
\end {itemize}

In a previous paper \cite {Car2} devoted to the analysis of the occurrence of blow-up it has been found that, in the case of harmonic/inverted potential, the 
momentum of inertia $\I$ satisfies to the following equation 
\bee
\frac {d^2 \I}{d t^2} + {8 \alpha }\I =  C_{\I} +  {2\nu} \frac {\mu -2}{\mu +1} \| \psi \|_{L^{2\mu +2}}^{2\mu +2} \, , \ 
C_{\I} =  {4 {\mathcal E} (\psi_0 )} \, . \label {Eq19}
\eee

As in the free case we consider now the equation for the variance $\V$.

\begin {lemma} \label {Lemma1} The variance $\V$ satisfies to the following equation
\bee
\frac {d^2 \V }{dt^2} + 8 \alpha  \V = C_{\V} + 
{2\nu} \frac {\mu -2}{\mu+1} \| \psi \|_{L^{2\mu +2}}^{2\mu +2}\, , \label {Eq20}
\eee
where
\bee
C_{\V} = - 2  {\hat p_0^2} -  4 \alpha \hat x_0^2 +C_{\I} \, . \label {Eq21}
\eee
\end {lemma}

\begin {proof}
Indeed, form (\ref {Eq19}) it turns out that the variance is a solution to the equation
\be
\frac {d^2 \V}{dt^2} + 8 \alpha \V = - \frac {d^2 \langle \hat x \rangle^2}{dt^2} - 8 \alpha \langle \hat x \rangle^2 +   C_{\I}+ 
 {2\nu} \frac {\mu -2}{\mu+1} \| \psi \|_{L^{2\mu +2}}^{2\mu +2}\, , 
\ee
where $\langle \hat x \rangle$ simply denotes $\langle \hat x \rangle^t$ and it is given by (\ref {Eq17}) (resp. (\ref {Eq18})) when $\alpha >0$ (resp. $\alpha <0$). \ We 
may remark that the term
\be
C = - \frac {d^2 \langle \hat x \rangle^2}{dt^2} - 8 \alpha \langle \hat x \rangle^2 
\ee
is constant. \ Indeed,
\be
C = -2 \left ( \frac {d \langle \hat x \rangle}{dt} \right )^2 - 2 \langle x \rangle \frac {d^2 \langle \hat x \rangle}{dt^2} - 8 \alpha \langle x \rangle^2 
= -2 \left ( \frac {d \langle \hat x \rangle}{dt} \right )^2 - 4 \alpha \langle x \rangle^2 
\ee
since $\frac {d^2 \langle \hat x \rangle}{dt^2} = - {2\alpha} \langle \hat x \rangle$ and thus
\be
\frac {dC}{dt} = -4 \frac {d \langle \hat x \rangle}{dt} \frac {d^2 \langle \hat x \rangle}{dt^2}  - {8} \alpha  \langle \hat x \rangle 
\frac {d \langle \hat x \rangle}{dt}   =0 \, .
\ee
Hence,
\be
C = - 2 \left ( \frac {d\langle \hat x \rangle }{dt} \right )^2_{t=0} - {4}\alpha \hat x_0^2 = - 2 {\hat p_0^2} - {4} \alpha \hat x_0^2 \,  
\ee
and (\ref {Eq20}) follows. 
\end {proof}

We recall that the initial condition associated to (\ref {Eq19}) and (\ref {Eq20}) are 
\bee
\V_0 := \V (0) =  \I_0 - \hat x_0^2 =  \| \hat x \psi_0 \|^2 - \hat x_0^2 \label {Eq22}
\eee
and 
\bee
\dot {\V}_0  := \frac {d\V (0)}{dt} = {2} \left [ \Re \langle \hat x \psi_0 , \hat p \psi_0 \rangle  - \hat x_0 \hat p_0 \right ] \, . \label {Eq23}
\eee

Let us consider now the differential equation (\ref {Eq20}) for $\mu \ge 2$ and $\nu <0$. \ From Lemma \ref {Lemma3} in Appendix \ref {AppB} we have 
that $0 \le \V (t) \le \z (t)$ where $\z (t)$ is the solution to 
\be
\left \{
\begin {array}{l}
\frac {d^2 \z}{dt^2} + {8\alpha}  \z = C_{\V}  \\
\z (0)= \V_0 \ \mbox { and } \ \frac {d \z (0)}{dt}=\dot \V_0 
\end {array}
\right. \, . 
\ee
If we set $\Omega = 2\lambda = \sqrt {8|\alpha |}$ then the solution $\z (t)$ is given by 
\be
\z(t) =
\left \{
\begin {array}{ll}
\z_H (t) := \frac {\dot \V_0}{\Omega} \sin (\Omega t) + \V_0 \cos (\Omega t) + \frac {1}{\Omega^2} C_\V \left [ 1 - \cos (\Omega t) \right ] & \, , 
\ \mbox { if } \ \alpha >0 \\
 & \\
\z_I (t) := \frac {\dot \V_0}{\Omega} \sinh (\Omega t) + \V_0 \cosh (\Omega t) -\frac {1}{\Omega^2} C_\V \left [ 1 - \cosh (\Omega t) \right ] & \, , 
\ \mbox { if } \ \alpha <0
\end {array}
\right. \,  .
\ee

Now, we are ready to apply the enhanced Zakharov-Glassey method.

\subsection {Harmonic oscillator - Criterion for blow-up} \label {Sez5_1} In the case of the harmonic oscillator potential, where $\alpha >0$, from Lemma \ref {Lemma3} 
in Appendix \ref {AppB} it follows that the variance $\V (t)$ is bounded from above by the function 
\be
\z_H(t) = \sqrt {a^2+b^2} \sin (\Omega t + \varphi ) + c 
\ee
for any $t$ such that $\Omega |t| \le \pi$, where $\varphi $ is a phase term such that
\be
\frac {a}{\sqrt {a^2+b^2}} = \cos \varphi \, , \ \frac {b}{\sqrt {a^2+b^2}} = \sin \varphi \, , \ a:= \frac {\dot \V_0}{\Omega} \, , \ b := \V_0 - 
\frac {C_\V}{\Omega^2} \, , \ c := \frac {C_\V}{\Omega^2} \, . 
\ee
Since $\z_H (\pm \pi/\Omega )= \frac {2}{\Omega^2}C_\V - \V_0$ then we have blow-up in the future and in the past if 
\bee
2\frac {C_\V}{\Omega^2} - \V_0 \le 0 \, . \label {Eq24Bis}
\eee
If not, since by means of a straightforward calculation it follows that 
\be
\min_{t \in [-\pi /\Omega , +\pi /\Omega ]} \V (t ) \le \min_{t \in [-\pi /\Omega , +\pi /\Omega ]} \z_H (t ) = \frac {C_{\V}}{\Omega^2} - 
\sqrt {\frac {\dot \V_0^2}{\Omega^2} + \left ( \V_0 - \frac {C_{\V}}{\Omega^2}\right )^2}
\ee
and then there exists blow-up in the past or in the future if
\bee
\dot \V_0^2 + \V_0^2 \Omega^2 - 2 \V_0 C_{\V} \ge 0 \, . \label {Eq24}
\eee

Thus, we have proved the following results.

\begin {theorem} \label {Teo3}
Let $\psi_0 \in \Sigma $ be the normalized initial wavefunction; let $\mu \ge 2$, $\alpha >0$ and $\Omega = \sqrt {8\alpha }$; let $C_\V$, $\V_0$ and $\dot \V_0$ 
defined as in (\ref {Eq21}), (\ref {Eq22}) and (\ref {Eq23}). \ Then, in the focusing nonlinearity case such that $\nu <0$  blow-up occurs in the 
past \emph { and } in the future at some instants $\tilde T_- \le t_-^\star <0 < t_+^\star \le \tilde T_+$ if (\ref {Eq24Bis}) is satisfied; where $\tilde T_\pm$ 
are the solutions to the equation $\z_H (t)=0$ in the interval $[-\pi /\Omega , \pi /\Omega ]$. \ If (\ref {Eq24Bis}) is not satisfied, but (\ref {Eq24}) holds true 
then blow-up occurs in the past \emph { or } in the future in the interval $[-\pi/\Omega , +\pi /\Omega ]$. 
\end {theorem}

\begin {remark} \label {Nota8}
We compare now the results above with the ones  given by Proposition 3.2 \cite {Car1} where occurrence of blow-up in the past and in the future was proved in 
the case of harmonic potential where $\alpha >0$, focusing nonlinearity where $\nu <0$, and under the conditions $\mu \ge 2$ and 
\bee
\frac {1}{2} \| \nabla \psi_0 \|^2_{L^2} + \frac {\nu}{\mu +1} \| \psi_0 \|^{2\mu +2}_{L^{2\mu +2}} \le 0 \, . \label {Eq25}
\eee
In fact, condition (\ref {Eq25}) implies that (since $\V_0 \ge 0$)
\be
\Omega^2 \I_0 \ge {8} {\mathcal E} 
& \Leftrightarrow &  \Omega^2 \V_0 - 2 C_{\V}  \ge  {4} \hat p_0^2  \, . 
\ee
That is, if (\ref {Eq25}) occurs then (\ref {Eq24Bis}) is satisfied (but not vice versa). 
\end {remark}

\subsection {Inverted oscillator - Criterion for blow-up} \label {Sez5_2} In the case of the inverted oscillator potential where $\alpha <0$ a similar argument proves 
that the variance $\V (t)$ is bounded from above by the function $\z_I (t)$ for any $t\in \R$. \ Then, it follows that

\begin {theorem}\label {Teo4}
Let $\psi_0 \in \Sigma $ be the normalized initial wavefunction; let $\mu \ge 2$, $\alpha <0$ and $\Omega = \sqrt {8|\alpha |}$; let $C_\V$, $\V_0$ and $\dot \V_0$ 
defined as in (\ref {Eq21}), (\ref {Eq22}) and (\ref {Eq23}). \ Let it now
\be
a:= \frac {\dot \V_0}{\Omega}\, , \ b:=  \V_0 + \frac {C_{\V}}{\Omega^2} \ \mbox { and } \ c := - \frac {C_{\V}}{\Omega^2} \, . 
\ee
Then, in the focusing nonlinearity case such that $\nu <0$ blow-up occurs if

\begin {itemize}

\item [i.] $b<-|a|$; in that case blow-up occurs in \emph {both the past and in the future}.

\item [ii.] $|a|<b$ and $\sqrt {b^2-a^2} + c \le 0$; in that case blow-up occurs \emph {only in the future} (if $a<0$) or \emph {only in the past} (if $a>0$).

\item [iii.] $|a| >|b|$; in that case we have blow-up in \emph { the past } if $a>0$ \emph { or in the future } if $a<0$.

\item [iv.] $|a|=|b|$; in that case we have blow-up if $bc<0$, in particular we have blow-up \emph {in the past} if $a>0$ or \emph {in the future} if $a<0$.

\end {itemize}

\end {theorem}

\begin {proof}
Let us introduce the function $\z (\tau ) = \z_I ( t)$ where $\tau =\Omega t$, then 
\be
\z (\tau ) :=a \sinh ( \tau )+b \cosh (\tau ) +c \, , \  \z (0) = \V (0) >0 \, , 
\ee
and where $a$, $b$ and $c$ are defined above. \ If

\begin {itemize}

\item [1)] $|a| < |b|$ then $\frac {d\z (\tau_1)}{d \tau} =0$ where $\tau_1 = \mbox {arctanh} \left (-\frac {a}{b} \right ) $. \ In particular, if:

\begin {itemize}

\item [1a)] $b< 0$ then $\lim_{\tau \to \pm \infty} \z (\tau ) = - \infty$ and thus there exists $\T_- < 0 < \T_+$ such that $\z ( \T_\pm )=0$. \ In such a case 
we have blow-up in the past \emph { and } in the future.

\item [1b)] $0  < b$ then $\lim_{\tau \to \pm \infty} \z (\tau ) = + \infty$. \ We compute now 
\be
\z (\tau_1 )= \sqrt {b^2-a^2} + c \, .
\ee
Thus, if 
\be
\sqrt {b^2-a^2} + c \le 0 
\ee
then we have blow-up in the future if $a<0$ \emph { or } in the past if $a>0$. 

\end {itemize}

\item [2)] $|a| > |b|$ then $\z (\tau )$ is a monotone increasing (resp. decreasing) function if $a>0$ (resp. $a<0$) such that $\lim_{\tau \to \pm \infty} \z (\tau ) 
= \pm \infty $ (resp. $\mp \infty$); therefore there exists $\T_- < 0 $ (resp. $0<\T_+$) such that $\z ( \T_- )=0$ (resp. $\z (\T_+ )=0$), and thus we have blow-up 
in the past (resp. in the future).

\item [3)] $a=b$ then $\frac {d\z (\tau )}{d\tau } \not= 0$ for any $\tau$. \ Hence, if:

\begin {itemize}

\item [3a)] $a>0$ then $\frac {d\z (0 )}{d\tau } >0$ and then $\frac {d\z (\tau )}{d\tau } >0$ for any $\tau$; furthermore, $\lim_{\tau \to - \infty} \z (\tau )= c$ 
and $\lim_{\tau \to + \infty} \z (\tau ) = +\infty$. \ Thus, if $c <0$ then there exists $\T_- < 0 $ such that $\z ( \T_- )=0$ and so we have blow-up in the past.

\item [3b)] $a<0$ then $\frac {d\z (0 )}{d\tau } <0$ and then $\frac {d\z (\tau )}{d\tau } <0$ for any $\tau$; furthermore, $\lim_{\tau \to - \infty} \z (\tau )= c$ 
and $\lim_{\tau \to + \infty} \z (\tau ) = - \infty$. \ Thus, if $c > 0$ then there exixts $0 < \T_+ $ such that $\z ( \T_+ )=0$ and so we have blow-up in the future.

\end {itemize}

\item [4)] $a=-b$ then $\frac {d\z (\tau )}{d\tau } \not= 0$ for any $\tau$. \ Hence, if:

\begin {itemize}

\item [4a)] $a>0$ then $\frac {d\z (0 )}{d\tau } >0$ and then $\frac {d\z (\tau )}{d\tau } >0$ for any $\tau$; furthermore, $\lim_{\tau \to - \infty} 
\z (\tau )= -\infty$ and $\lim_{\tau \to + \infty} \z (\tau ) = c$. \ Thus, if $c >0$ then there exists $\T_- < 0 $ such that $\z ( \T_- )=0$ and so we have 
blow-up in the past.

\item [4b)] $a<0$ then $\frac {d\z (0 )}{d\tau } <0$ and then $\frac {d\z (\tau )}{d\tau } <0$ for any $\tau$; furthermore, $\lim_{\tau \to - \infty} 
\z (\tau )= +\infty $ and $\lim_{\tau \to + \infty} \z (\tau ) = c$. \ Thus, if $c <0$ then there exists $\T_- < 0 $ such that $\z ( \T_- )=0$ and so we have 
blow-up in the past.

\end {itemize}

\end {itemize}

Collecting all these results then 
Theorem \ref {Teo4} follows.
\end {proof}

\begin {remark} \label {Nota9}
We compare now the results above with those given by Theorem 1.1 \cite {Car3}. \ For example, \cite {Car3} proved that in the case of inverted potential, where $\alpha <0$, and focusing nonlinearity, where $\nu <0$, under the condition $\mu \ge 2$ and
\bee
\frac {1}{2} \| \nabla \psi_0 \|^2_{L^2} + \frac {\nu}{\mu +1} \| \psi_0 \|^{2\mu +2}_{L^{2\mu +2}} < -|\alpha | \| x \psi_0 \|^2_{L^2} - \sqrt {2|\alpha|} \left | \Re \langle \hat x \psi_0 , \hat p \psi_0 \rangle \right | \label {Eq26}
\eee
then blow-up occurs in the future \emph {and} in the past at some instant. \ By means of a straightforward calculation it can be proved that if condition (\ref {Eq26}) is satisfied, then condition i. of Theorem \ref {Teo4} holds true, but not vice versa.
\end {remark}

\appendix

\section {Functional inequalities} \label {AppA}

\begin {lemma} \label {Lemma2}
The following inequality holds true: let $y \in \R$ and let 
\be
\Gamma := \Gamma (y) = \langle f , (x-y)^2 f \rangle_{L^2} 
\ee
for any test function $f\in L^2 (\R , dx) $ such that $xf \in L^2 (\R , dx)$. \ Then, for any $q\ge 0$:
\bee
\| f \|^{2q+2}_{L^{2q+2}} &\le & C \sqrt {\Gamma} \| f \|_{L^2}^q \| f' \|_{L^2}^{q+1} \, , 
\label {Eq27}
\eee
for some positive constant $C$, where $f' = \frac {df}{dx}$.
\end {lemma}

\begin {proof}
Indeed:
\be
\| f \|^{2q+2}_{L^{2q+2}} 
&=& \int_{\R} \frac {\partial (x-y)}{\partial x} f^{q+1}\bar f^{q+1} dx \\ 
&=& -  (q+1) \int_{\R}(x-y) |f|^{2q} \left [f' \bar f+ f \bar f' \right ] dx \, . 
\ee
Hence
\be
\| f \|^{2q+2}_{L^{2q+2}}  \le 2(q+1)  \| (x-y )f \|_{L^2} \| f\|_{L^\infty}^{2q} \| f' \|_{L^2}\, .
\ee
Now, recalling that from the Gagliardo-Nirenberg inequality one has that 
\be
\| f \|_{L^\infty} \le \sqrt {2} \| f' \|^{1/2}_{L^2}
\|  f \|^{1/2}_{L^2}
\ee
then it follows that 
\be
\| f \|^{2q +2}_{L^{2q +2}}  \le  C \sqrt {\Gamma }
\|  f \|_{L^2}^{q} \| f' \|_{L^2}^{q+1}
\, ,
\ee
for some positive constant $C$.
\end {proof}

\begin {corollary} \label {Coro3} In particular, if $y=0$ and $q=0$ then we have that for some positive constant $C$:
\bee
\| f \|^{2}_{L^{2}} \le C \sqrt {\I}
 \| f' \|_{L^2} \, , 
\label {Eq28}
\eee
where $ \Gamma (0) = \| x f \|_{L^2}^2 = \I $ is the moment of inertia; if $y = \langle \hat x \rangle = \langle f , x f \rangle_{L^2} $ and $q=0$ then we have that:
\bee
\| f \|^{2}_{L^{2}} \le C  \sqrt {\V}
\| f' \|_{L^2} \, , 
\label {Eq29}
\eee
where $ \Gamma (\langle \hat x \rangle ) = \| (x - \langle \hat x \rangle ) f \|_{L^2}^2 = \V$ is the variance.
\end {corollary}

\section {Comparison between solutions of the harmonic/inverted oscillator} \label {AppB}

Let $\V_{\pm} (t) $ be the solution to the differential equation 
\bee
\left \{
\begin {array}{l}
\frac {d^2 \V_\pm}{dt^2} \pm \Omega^2 \V_\pm = C + f(t) \\
\V_\pm (0) = \V_{\pm ,0} \ \mbox { and } \ \frac {d \V_\pm (0)}{dt}  = \dot \V_{\pm ,0}
\end {array}
\right. \, , \label {Eq30}
\eee
where $C$ is a constant factor and $f(t) \le 0$ for any $t$; and let $\z_\pm (t) $ be the solution to the differential equation 
\bee
\left \{
\begin {array}{l}
\frac {d^2 \z_\pm}{dt^2} \pm \Omega^2 \z_\pm = C  \\
\z_\pm (0) = \V_{\pm ,0} \ \mbox { and } \ \frac {d \z_\pm (0)}{dt} = \dot \V_{\pm ,0}
\end {array}
\right. \, . \label {Eq31}
\eee
Then, the difference $\Z_\pm (t) = \V_{\pm} (t) - \z_\pm (t)$ solves the differential equation 
\be
\left \{
\begin {array}{l}
\frac {d^2 \Z_\pm}{dt^2} \pm \Omega^2 \Z_\pm = f(t)  \\
\Z_\pm (0) = 0 \ \mbox { and } \ \frac {d \Z_\pm (0)}{dt} = 0
\end {array}
\right. \, .
\ee
Hence, we have that
\be
\Z_+ (t) = \frac {1}{\Omega} \int_0^t \sin \left [ \Omega (t-s) \right ] f(s) ds \le 0 \ \mbox { if } \ \Omega |t| \le \pi \, 
\ee
and
\be
\Z_- (t) = \frac {1}{\Omega} \int_0^t \sinh \left [ \Omega (t-s) \right ] f(s) ds \le 0 \, , \ \forall  t \in \R \, . 
\ee
In conclusion,

\begin {lemma} \label {Lemma3}
Let $\V_{\pm}$ be the solution to (\ref {Eq30}), and let
\be
\z_+ (t) = \frac {\dot \V_{+,0}}{\Omega} \sin (\Omega t) + \V_{+,0} \cos (\Omega t) + \frac {1}{\Omega^2} C_\z \left [ 1 - \cos (\Omega t) \right ] 
\ee
and
\be
\z_- (t) = \frac {\dot \V_{-,0}}{\Omega} \sinh (\Omega t) + \V_{-,0} \cosh (\Omega t) -\frac {1}{\Omega^2} C_\z \left [ 1 - \cosh (\Omega t) \right ] 
\ee
be the solution to (\ref {Eq31}). \ Then
\be
\V_+ (t) \le \z_+ (t) \, , \ \forall t \in \left [ - \frac {\pi}{\Omega} ,+ \frac {\pi}{\Omega}\right ]
\ee
and
\be
\V_- (t) \le \z_- (t) \, , \ \forall t \in \R \, . 
\ee
\end {lemma}

\section {A formal touch - the virial identity} \label {AppC}
Here we formally derive the virial identity for any real-valued potential $V(x)$. 

Hereafter, we denote $\psi_t $ by $\psi$ and $\psi' = \frac {\partial \psi}{\partial x}$, $\psi'' = \frac {\partial^2 \psi}{\partial x^2}$, 
$\dot \psi = \frac {\partial \psi}{\partial t}$, $\dot \I = \frac {d\I }{dt}$, $\ddot \I = \frac {d^2\I }{dt^2}$, and so on. 

Let (\ref {Eq3}) be the energy integral of motion (here we make no assumptions about the values of the mass $m$ and of the Planck constant $\hbar$): 
\be
\En (\psi ) := 
\frac {\hbar^2}{2m} \left \langle \psi' ,  \psi' \right \rangle + \langle \psi , V \psi \rangle +  
\frac {\nu}{\mu +1} \| \psi \|_{L^{2\mu +2}}^{2\mu +2}  \, ,  
\ee
Let
\be
\I(t)= \langle \hat x^2 \rangle^t = \langle \psi_t, x^2 \psi_t \rangle_{L^2}
\ee
be the momentum of inertia. \ It satisfies to the following \emph {virial identity}: 
\bee
\frac {d^2\I}{dt^2} =\frac {4}{m} \En - \frac {2}{m} \left [ \langle  \psi , x V' \psi \rangle +2 \langle \psi , V \psi \rangle \right ] + 
\frac {2\nu (\mu -2)}{m(\mu +1)} \| \psi \|_{L^{2\mu +2}}^{2\mu +2} . \label {Eq32}
\eee

In order to compute the derivatives of $\I (t)$ from (\ref {Eq7}) it follows that
\be
\dot \I  =  \frac {i}{\hbar}   \langle \psi , [H,\hat x^2] \psi \rangle
\ee
since $[|\psi |^{2\mu},\hat x^2 ]=0$. \ From this fact and since 
\be
[H,\hat x^2]\psi=  -\frac {\hbar^2}{2m} \left ( 2\psi + 4 x \psi' \right )
\ee
then 
\bee
\dot \I = - i \frac {\hbar}{m} \| \psi \|^2 - 2 i \frac {  \hbar}{m} \langle  x \psi , \psi ' \rangle \, . \label {Eq33}
\eee
From equation (\ref {Eq33}) and since the norm $\| \psi \|$ is a constant function with respect to the time then
\be
\ddot \I 
= - 2 i \frac {  \hbar}{m} \langle x \dot \psi , \psi ' \rangle - 2 i \frac {  \hbar}{m} \langle x \psi , \dot \psi ' \rangle =   2 i \frac {  \hbar}{m} \langle  \psi , \dot \psi  \rangle + 4 \frac {  \hbar}{m} \Im \left [ \langle  x\dot \psi ,  \psi'  \rangle \right ] 
\ee
where
\be
\langle  \psi , \dot \psi  \rangle  
= \frac {i}{\hbar} \langle \psi, H\psi +\nu |\psi |^{2\mu} \psi \rangle = -\frac {i}{\hbar} {\mathcal E} - \frac {i}{\hbar} \frac {\nu \mu}{\mu+1} \| \psi \|_{L^{2\mu +2}}^{2\mu +2}
\ee
because $\dot \psi = -\frac {i}{\hbar}H\psi -i\frac {\nu}{\hbar}|\psi |^{2\mu } \psi$, and 
\be
\langle  x \dot \psi , \psi ' \rangle 
=  -\frac {i}{\hbar} \langle  H\psi +\nu |\psi |^{2\mu} \psi , x \psi ' \rangle = 
\frac {i}{\hbar} B + \frac {i\nu }{\hbar}  A \, , 
\ee
where
\be
B= \langle  H\psi ,  x \psi ' \rangle \ \mbox { and } \ A= \langle |\psi |^{2\mu} \psi ,  x \psi ' \rangle  \, . 
\ee
By integrating by parts then 
\be
A &=& \int_{\R} x \bar \psi^{\mu +1}  \psi^\mu  \psi ' dx \\ 
&=& - \int_{\R} \psi^{\mu +1} \bar \psi^{\mu +1} dx
- (\mu+1) \int_{\R} x \bar \psi^{\mu}  \psi^{\mu+1} \bar \psi ' dx - \mu \int_{\R} x \bar \psi^{\mu +1}  \psi^\mu  \psi ' dx \\
&=& -  \| \psi \|_{L^{2\mu +2}}^{2\mu +2} - (\mu +1) \bar A - \mu A
\ee
from which it follows that 
\be
(A+\bar A) = - \frac {1}{\mu +1}  \| \psi \|_{L^{2\mu +2}}^{2\mu +2} \, .
\ee
Now, let
\be
B =  B_1 + B_2 \ \mbox { where } \ 
B_1 = - \frac {\hbar^2}{2m} \langle \psi'' ,  x \psi' \rangle \ \mbox { and } \ 
B_2 = \langle V \psi ,  x \psi' \rangle \, . 
\ee
A straightforward calculation yields to  
\be
B_2 =   - \langle V \psi' , x \psi \rangle  - \langle(x V)' \psi ,  \psi \rangle = - \bar B_2 - \langle(x V)' \psi ,  \psi \rangle \, , 
\ee
hence
\be
(B_2 + \bar B_2 ) =  - \langle(x V)' \psi ,  \psi \rangle \, . 
\ee
Similarly
\be
B_1 
&=& - \frac {\hbar^2}{2m} \langle \psi'' , x \psi' \rangle = \frac {\hbar^2}{2m} \langle \psi' , \hat x \psi'' \rangle + \frac {\hbar^2}{2m} \langle \psi' ,  \psi' \rangle \\
&=& -\bar B_1 + {\mathcal E} - \langle \psi , V \psi \rangle - \frac {\nu}{\mu+1} \| \psi \|_{L^{2\mu+2}}^{2\mu +2}
\ee
from which follows that 
\be
(B_1 + \bar B_1) =  {\mathcal E} - \langle \psi , V \psi \rangle- \frac {\nu}{\mu+1} \| \psi \|_{L^{2\mu+2}}^{2\mu +2}\, . 
\ee
In conclusion:
\be
\ddot \I 
&=& 2 i \frac {  \hbar}{m} \left [ -\frac {i}{\hbar} {\mathcal E} - \frac {i}{\hbar} \frac {\nu \mu}{\mu+1} \| \psi \|_{L^{2\mu +2}}^{2\mu +2} \right ] + 4 \frac {  \hbar}{m} \Im \left [ \frac {i}{\hbar} A + \frac {i\nu }{\hbar} B \right ] \\
&=& \frac {2}{m} \En + \frac {2\nu \mu}{m(\mu+1)} \| \psi \|_{L^{2\mu +2}}^{2\mu +2} + \frac {4}{m} \Re \left [ B + \nu A \right ] \\ 
&=& \frac {4}{m} \En - \frac {2}{m} \left [ \langle  \psi , x V' \psi \rangle +2 \langle \psi , V \psi \rangle \right ] + \frac {2\nu (\mu -2)}{m(\mu +1)} \| \psi \|_{L^{2\mu +2}}^{2\mu +2}
\ee
Thus (\ref {Eq32}) follows.

\begin {remark} \label {Nota10}
We remark that the virial identity (\ref {Eq32}) in the particular cases $V(x)\equiv 0$, $V(x) = \alpha x $ and $V(x) = \alpha x^2$, for $\alpha \in \R$, respectively becomes
\bee
\frac {d^2\I}{dt^2} = \frac {4}{m} \En  + \frac {2\nu (\mu -2)}{m(\mu +1)} \| \psi \|_{L^{2\mu +2}}^{2\mu +2} \, , \  \mbox { if } V(x) \equiv 0\, ,  \label {Eq34}
\eee
\bee
\frac {d^2\I}{dt^2} = \frac {4}{m} \En - \frac {6 \alpha}{m} \langle \hat x \rangle^t + \frac {2\nu (\mu -2)}{m(\mu +1)} \| \psi \|_{L^{2\mu +2}}^{2\mu +2} \, , \  \mbox { if } V(x) =\alpha x \, , \label {Eq35}
\eee
and 
\bee
\frac {d^2\I}{dt^2} &=& \frac {4}{m} \En - \frac {8\alpha}{m} \langle \hat x^2 \rangle^t + \frac {2\nu (\mu -2)}{m(\mu +1)} \| \psi \|_{L^{2\mu +2}}^{2\mu +2} \nonumber \\
&=& \frac {4}{m} \En - \frac {8\alpha}{m} \I + \frac {2\nu (\mu -2)}{m(\mu +1)} \| \psi \|_{L^{2\mu +2}}^{2\mu +2} \, , \  \mbox { if } V(x) =\alpha x^2 \, . \label {Eq36}
\eee
\end {remark}

\end{document}